\documentclass[10pt, twocolumn]{article}
\usepackage{multirow}
\usepackage{subfigure}
\usepackage{graphicx}
\usepackage{longtable}
\usepackage{courier}
\usepackage{float}
\usepackage{cite}
\usepackage{amsmath,amssymb}
\usepackage{bbold}
\usepackage{dsfont}
\usepackage[in]{fullpage}
\usepackage[english]{babel}
\usepackage{amsthm}
\usepackage{changes}

\newtheorem{theorem}{Theorem}
\newtheorem{corollary}{Corollary}[theorem]
\newtheorem{definition}{Definition}
\newtheorem*{remark}{Remark}
\newcounter{MYtempeqncnt}
\usepackage[T1]{fontenc}
\usepackage{authblk}
\usepackage{array}
\usepackage{url}

\usepackage[pdfborder={0 0 0}]{hyperref}
\usepackage[margin=0.55in]{geometry}
\newcommand{\ou}{%
  \mathrel{%
    \vcenter{\offinterlineskip
      \ialign{##\cr$<$\cr\noalign{\kern-1.5pt}$>$\cr}
			 }%
  }%
}	
 \vspace{-8ex}
  \date{}

\begin{document}

 \title{A Statistical Threshold for Adversarial Classification in Laplace Mechanisms}%\footnote{Further revisions of this paper is available on \url{https://arxiv.org/abs/2105.05610}}}
\author{%
Ay\c{s}e \"{U}nsal \qquad Melek \"{O}nen\\ EURECOM, France \\ \url{firstname.lastname@eurecom.fr}
 %\and 
%Melek \"{O}nen\\ Eurecom, France\\ \url{melek.onen@eurecom.fr}
}

  %\author*[1]{Ay\c{s}e \"{U}nsal}
%%
  %\author[2]{Melek \"{O}nen}
%%
  %\affil[1,2]{EURECOM, France E-mail: firstname.lastname@eurecom.fr}}
%%
  %\affil[2]{EURECOM, France E-mail: melek.onen@eurecom.fr}
%
  %\title{\huge A Statistical Threshold to Remain Undetected for Laplacian and Gaussian Differentially Private Mechanisms}
%
  %\runningtitle{A Statistical Threshold to Remain Undetected for Laplacian and Gaussian Differentially Private Mechanisms}
%
  %%\subtitle{...}
\maketitle
 \pagestyle{plain}

  \begin{abstract}
This paper studies the statistical characterization of detecting an adversary who wants to harm some computation such as machine learning models or aggregation by altering the output of a differentially private mechanism in addition to discovering some information about the underlying dataset. An adversary who is able to modify the published information from a differentially private mechanism aims to maximize the possible damage to the system while remaining undetected. We present a trade-off between the privacy parameter of the system, the sensitivity and the attacker's advantage (the bias) through determining the threshold for the best critical region of the hypothesis testing problem for deciding whether or not the adversary's attack is detected. Such trade-offs are provided for Laplace mechanisms using one-sided and two-sided hypothesis tests. Corresponding error probabilities are analytically derived and ROC curves are presented for various levels of the sensitivity, the absolute mean of the attack and the privacy parameter. Subsequently, we provide an interval for the bias induced by the adversary so that the defender detects the attack. Finally, we adapt the \textit{Kullback-Leibler differential privacy} to adversarial classification.

  %\begin{keywords}
	%differential privacy, hypothesis testing, Laplace dp noise, Gaussian dp noise
	%\end{keywords}
	
	\end{abstract}
\section{Introduction}

The widespread use of Big Data technologies has opened the door for malicious attacks resulting in potentially devastating consequences in critical applications such as autonomous driving or healthcare. In particular, an adversary may look for means to modify models or their outputs and consequently wreak havoc on a system and its users. Furthermore, such techniques usually rely on large datasets to be efficient which increases the chance of fraudulent use of personal information. \textit{Adversarial classification} (also called anomaly detection) is a statistical tool enabling the detection of modification/misclassification attacks whereas privacy-preservation commonly makes use of so-called differential privacy (DP) mechanisms. A mechanism, such as a randomized function of a dataset, is said to be \textit{differentially private} if the level of privacy of individual participants and the output of the mechanism remain unaltered even when any one of the participants decides to submit or equivalently remove his/her input from the statistical dataset. This definition is also applicable to aggregate information of all participants. This paper studies the security of systems where DP mechanisms are also used by adversaries. 

Differential privacy (DP) \cite{r8} is defined as a stochastic measure of privacy -that has a precise mathematical formulation- to ensure privacy of individual users when handling large datasets. DP mechanisms have furthermore been used to develop practical methods for protecting private user-data at the moment they provide information to the system. In these cases, the use of DP measure aims to maintain the accuracy of the underlying operation without incurring a cost of the privacy of individual participants.  In some sense, DP is a notion of robustness against changes in the dataset. The degree of this change is measured/determined by an adjustable privacy parameter and the amount of the change, that any single argument to the system reflects on its output, is called the sensitivity. 

Statistical classification is now widely used as a supervised machine learning approach and consists in placing or \textit{classifying} an item into one of several categories based on a number of measurements of interest. In \cite{r7}, classification is described as a hypothesis testing problem for choosing between two possible values that the parameter(s) of a probability distribution can take on to place this item into the right category. Adversarial classification is an application of statistical classification where an adversary tries to fool a classifier which detects outliers in order to remain undetected.

In this paper, we consider a scenario where privacy enhancing technologies, which were originally designed to support privacy protection of legitimate individuals, are used by adversaries to harm the security of systems. We assume that the adversary is aware of the underlying DP mechanism and its parameters and wants to benefit from it using it as an attack tool to avoid being detected \cite{r1, r12}. The adversary's goal is to maximize the possible damage while minimizing the probability of being detected. We study the statistical framework of adversarial classification with DP. Our goal is to evaluate the impact of privacy parameters on the actual power of the adversarial classification. In particular, we focus on the aggregation operation whereby parties contribute with some individual numerical data and the system collecting this data computes the sum of them. We establish a stochastic relation between the probability of the adversary's success and the privacy parameter in the specific case of Laplace mechanisms.

\subsection{Related Work and Methodology}
This part is reserved for a discussion on related work and background of the addressed problem emphasizing the differences between the existing literature and the current paper along with the methodology that is used.

The addressed problem in this work differs from existing work on DP which considers an adversary model where the goal of the attacker is to solely discover some information about the dataset. For instance, the assumption in \cite{r2} is that the adversary has the knowledge of the entire dataset except for one entry. This translates to the implicit strong adversary assumption. In this paper, our aim is to extend this model with a stronger adversary who also wants to harm the dataset and the output of the mechanism. We consider an adversary who is able to modify (add, replace, delete, etc.) the published information from a differentially private mechanism which is a noisy version of the output. The adversary's goal in this model is to maximize the possible damage while remaining undetected. Thus, there are two sides of what the adversary wants to achieve: (i) s/he gives false data by modifying, replacing, adding or changing the released information in any other way, with the biggest possible difference from the real data, (ii) all this modification has to be achieved without being detected. On the defender's end, the mechanism wants to preserve DP and detect the attack.

A simpler version of the described problem is addressed by \cite{r1} from an adversarial perspective and the two conflicting goals of the adversary is formulated as an optimization problem where maximizing the bias induced by the adversary is the objective function. But the privacy parameter does not take part in the formulation of this optimization problem in \cite{r1}. We seek a characterization of the trade-off between the attack (the change in the output induced by the adversary) and the privacy parameter. On the other hand, in \cite{r4}, the authors show that the sensitivity of a mechanism has also an impact on the differentially private output. The noise to be added on the output is calibrated accordingly as a function of the noise distribution. Such a characterization of the problem described in this paper introduces a third element as the value of the attack to be included in this adjustment of the DP noise with respect to (w.r.t.) the sensitivity of the system. This would allow us to be able to determine a threshold for detecting the attacker, alternatively, for the attacker to remain undetected. 

As for the methodology, we will use the framework of statistical hypothesis testing in a similar vein to \cite{r3} where the authors determine an appropriate value of the privacy parameter as a function of false alarm and mis-detection probabilities in deciding on the presence or absence of a particular record in a dataset. Similarly, in \cite{r13}, the author studies the differentially private hypothesis testing in the local setting where users locally add the DP noise on their personal data before submitting it to the dataset. In this paper, we tailor this approach for the problem described above as a first attempt for a solution for anomaly detection in Laplace mechanisms under global DP where the personal sensitive data is transmitted to a central server by the users and the server applies DP noise on the data before its release to the public. We propose an interval for the privacy budget of Laplace mechanisms so that the defender detects the attack using hypothesis testing.

In summary, we aim to extend the adversary model to an adversary who could significantly modify the released information without being detected. The ultimate goal is to address the trade-off between privacy, sensitivity and attacker's advantage in a statistical framework.

\subsection{Contributions}
The following list is a brief summary of the contributions of this paper.
\begin{itemize}
\item We consider a new attacker model whereby the adversary takes advantage of the underlying differentially private mechanism in order not to be detected in his/her attack.
\item We derive a trade-off between the privacy protected adversary's advantage and the security of the system for the adversary to remain undetected while giving as much damage as possible to the system or, alternatively, for the defender to preserve the privacy of the system and detect the attacker. This trade-off is defined in the framework of statistical hypothesis testing similarly to \cite{r3}.
\item We establish statistical thresholds for detecting the attack as a function of the error probabilities for Laplace mechanisms through one-sided and two-sided hypothesis tests. Subsequently, these thresholds are used for deriving intervals for the attack and/or the privacy budget to remain undetected as a function of the error probabilities and the sensitivity. 
\item We adopt the Kullback-Leibler DP definition of \cite{r2} to the addressed problem for adversarial classification in Laplace mechanisms and present numerical comparisons of different cases where the sensitivity of the system is less than, equal to and greater than the bias induced by the adversary on the published information.
%\item Our numerical evaluation results show that 
\end{itemize}

The outline of the paper is as follows. Section \ref{sec:pre_model} provides necessary preliminaries from DP literature along with the detailed definition of the addressed problem and performance criteria. Section \ref{sec:results} presents the main results on the statistical trade-off between the privacy budget, sensitivity and the impact of the attack along with the proofs. In Section \ref{sec:KL}, we derive the so-called Kullback-Leibler DP for adversarial classification followed by numerical evaluations of the obtained results in Section \ref{sec:numeric}. Lastly, in Section \ref{sec:conc}, we draw conclusions based on our findings. 

\section{Preliminaries and Model \label{sec:pre_model}}
In this part, we revisit certain notions from the existing literature on DP which will also be used in this paper. These preliminaries will be followed by a detailed definition of the addressed problem.
We begin with defining neighboring datasets and sensitivity of DP. 
\begin{definition}\label{eq:distance}
Any two datasets that differ only in one record are called neighbors \cite{r5}. For two neighboring datasets, the following equality holds
%Hamming (or $l_1$) distance between two datasets $x$ and $\tilde{x}$ is defined by %\cite{r5} as
\begin{equation} \label{eq:hamming}
d(x, \tilde{x})=1
\end{equation} where $d(.,.)$ denotes the Hamming (or $l_1$) distance between two datasets. 
\end{definition}

\begin{definition} \label{def:sens}
Global sensitivity $s$ \cite{r4} of a function (or a query) $q: D \rightarrow \mathbb{R}^k$ is the smallest possible upper bound on the distance between the images of $q$ when applied to two neighboring datasets, i.e. the $l_1$ distance is bounded by $\|q(x)-q(\tilde{x})\|_1 \leq s$.
\end{definition} Verbally, sensitivity of a DP mechanism is the smallest possible upper bound on the images of a query (a mapping function) for neighbors. Hence it is a function of the type of the query having an opposite relationship with the privacy. Higher sensitivity of the query refers to a stronger requirement for privacy guarantee, consequently more noise is needed to achieve that guarantee.

\begin{definition}\label{def:dp}
$(\epsilon, \delta)-$ DP \cite{r5}:
A randomized algorithm $\mathcal{Y}$ is $(\epsilon, \delta)-$ differentially private if $\forall S \subseteq Range(\mathcal{Y})$ and for all neighboring datasets $x$ and $\tilde{x}$ within the domain of $\mathcal{Y}$ the following inequality holds.
\begin{equation}\label{ineq:dp}
\Pr\left[\mathcal{Y}(x) \in S\right] \leq \Pr\left[\mathcal{Y}(\tilde{x}) \in S \right] \exp\{\epsilon\}+\delta
\end{equation}
\end{definition}
Finally, we remind the reader of the Laplace distribution and Laplace DP mechanism that we will frequently make use of throughout the paper.
A differentially private system is named after the probability distribution of the perturbation applied onto the query output in the global setting. %We separate our results in two main groups for Laplacian $(\epsilon, 0)$ and Gaussian  $(\epsilon,\delta)$ differential privacy. 
The Laplace distribution, also known as the double exponential distribution, is defined as
\begin{equation}
Lap(x; \mu, b)= \frac{1}{2b}\exp\left\{-\frac{|x-\mu|}{b}\right\}
\end{equation} with the location parameter equal to its mean $\mu$ and variance $2b^2$ where $b$ denotes the scale parameter.
\begin{definition}
Laplace mechanism \cite{r4} is defined for a function (or a query) $q: D \rightarrow \mathbb{R}^k$ as follows
\begin{equation}
\mathcal{Y}(x, q(.), \epsilon)= q(x)+(Z_1, \cdots, Z_k)
\end{equation} where $Z_i \sim Lap(b=s/\epsilon)$, $i=1,\cdots, k$ denote i.i.d. Laplace random variables.
\end{definition}

\subsection{Problem Definition \label{subsec:model}}

In this section, we  provide a detailed description of the addressed problem and describe the quantitative components for establishing a statistical threshold for detecting the attack. Imagine consider the following scenario. A differentially private mechanism adds Laplace noise $Z$ on the query output $q(x)= \sum_{i=1}^{n} X_i$ which uses the dataset in the following form $\mathbf{X}= \{X_1, \cdots, X_n \}$. The noisy output is denoted by $Y_0$ and defined as $\mathcal{Y}(x, q(.), \epsilon)=Y_0= \sum_{i=1}^{n}X_i+ Z$. An adversary modifies this public information -which has been released by the server- by adding one extra record that we denote by $X_a$. Here the addition is applied onto the existing dataset without any constraint on the value of $X_a$, i.e. it could take up on a positive as well as a negative value. The modified output
becomes $\left(\sum_{i=1}^{n}X_i+X_a \right)+Z$.% following the probability distribution $Y_1 \sim p_{1}$. 

We define the following (simple) hypothesis testing problem in order to determine the threshold for deciding whether or not the adversary's attack is detected. 
\begin{equation}
\begin{aligned}\label{eq:ht}
H_{0} &: \textrm{defender\;does\;not\;detect}\; X_a  \\
H_{1} &: \textrm{defender\;detects}\; X_a 
\end{aligned}
\end{equation} The hypothesis testing problem defined above in (\ref{eq:ht}) can be translated into deciding on the DP noise distribution with its parameters. Here $H_{0}$ and $H_1$ correspond to DP noise following the probability distributions $p_0$ with mean $\mu_0$ and $ p_{1}$ with mean $\mu_1$, respectively. Therefore, the decision boils down to choosing between $Y_0-\sum_{i=1}^{n}X_i=Z$ and %$\mathbf{Y}_{1}-\sum_{i=1}^{n}X_i=\mathbf{Z}$ (or equivalently 
$Y_0-\left[\sum_{i=1}^{n}X_i+ X_a \right]$. Hence the shift in the location due to the addition of $X_a$ to the dataset is $\Delta \mu=\mu_1-\mu_0$. The corresponding likelihood ratio for this problem yields
\begin{equation}\label{likelihood_a-gen}
\Lambda = \frac{\mathcal{L}(p_1)}{\mathcal{L}(p_0)} \underset{H_1}{\overset{H_0}{\ou}} \kappa
\end{equation} where $\mathcal{L}(.)$ denotes the likelihood function for the corresponding hypothesis and $\kappa$ is some positive number to be determined. %and $(\mu_i, b_i)$ for $i=0,1$ represent the location and scale parameters of the distributions to be tested. 
 Such a threshold defines the critical region in statistical hypothesis tests where the null hypothesis is rejected. 

This paper presents a precise trade-off between the attacker's advantage (or the bias induced by the adversary) $\Delta \mu$, the sensitivity $s$ and the privacy parameter $\epsilon$ of the differentially private mechanism to characterize the threshold for rejecting the null hypothesis, i.e. detecting the attack, as a function of the error probabilities. The upcoming subsection provides a detailed definition of possible error events and their corresponding probabilities. 

\subsection{Performance Criteria in Hypothesis Testing}

$\alpha$ and $\beta$ respectively denote type I and type II error probabilities which are defined for the hypothesis testing problem in (\ref{eq:ht}) as follows:
\begin{align}
P_{FA}=\alpha &= \Pr\left[H_0\; \textrm{reject}| H_0 \; \textrm{is}\;\textrm{true}\right] \label{eq:alpha_def} \\
P_{MD}=\beta &=\Pr\left[H_1 \textrm{reject}| H_1 \textrm{is}\;\textrm{true}\right].\label{eq:beta_def}
\end{align}
Based on the definition of $\alpha$, also called the \textit{probability of false-alarm}, we denote its complement by $\bar{\alpha}=1-\alpha$. Similarly, due to (\ref{eq:beta_def}), the complement of type II error probability (or the \textit{probability of mis-detection}) is denoted by $\bar{\beta}=1-\beta$. The probability of detection $\bar{\beta}$  (i.e. correctly deciding $H_1$) is also called the \textit{power of the test} in the statistics or the recall in machine learning terminology. %Throughout this paper, we will refer to correctly deciding the alternative hypothesis $\bar{\beta}$ as the power of the test.

According to the Neyman-Pearson Theorem \cite{r6}, the likelihood ratio compared against some positive integer defines the best critical region of size $\alpha$ for testing a simple hypothesis against an alternative simple hypothesis with the largest (or equally largest) power of the test. An extension of this result to testing against a composite alternative hypothesis is also possible. Such an extension is called \textit{uniformly most powerful test} since for a test with the best critical region of size $\alpha$ is conducted for each possible value of the alternative hypothesis. 
Once we define the critical region for deciding between $H_0$ and $H_1$ in (\ref{eq:ht}) as a function of $\Delta \mu$, the privacy parameter $\epsilon$ and the sensitivity $s$, we will derive the error probabilities and the power of the test analytically as well as compute and depict them numerically.

\section{Main Results \label{sec:results}}

We separate our results in two main groups for $(\epsilon, 0)$-DP in Laplace mechanisms for one-sided and two-sided hypothesis tests.
%%%%%%%%%%%%%%%%%%%%%%%%%%%%%%%%%%%%%%%%%%%%%%%%%%%%%%%%%%%%%%%%%%%%%%%%%%%%%%%%%%%%%%%%%%%%%%%%%%%%%%%%%%%%%%%%%%%%%%%%%
%%%%%%%%%%%%%%%%%%%%%%%%%%%%%%%%%%%%%%%%%%%%%%%%%%%%%%%%%%%%%%%%%%%%%%%%%%%%%%%%%%%%%%%%%%%%%%%%%%%%%%%%%%%%%%%%%%%%%%%%%

\subsection{One-sided test}
In this part, we will investigate both cases setting the alternative hypothesis $H_1$ as either $\mu_1>\mu_0$ (i.e. $\Delta \mu >0$) or $\mu_1<\mu_0$ (i.e. $\Delta \mu <0$). This corresponds to a one-sided hypothesis testing problem. The decision of choosing between the hypotheses in (\ref{eq:ht}) boils down to choosing between $Y_0-\sum_{i=1}^{n}X_i=Z \sim \mathrm{Lap}(z;\mu_0, s/\epsilon)$ and $Y_0-\left[\sum_{i=1}^{n}X_i+ X_a \right]=Z \sim \mathrm{Lap}(z;\mu_1, \theta (s/\epsilon))$ where $\theta \geq1$ as the measure of the change in the privacy budget of the system whereas $s$ and $\epsilon$ denote the sensitivity and privacy parameter, respectively. It should be noted that setting $\theta=1$ translates the hypothesis test in (\ref{eq:ht}) into testing only the location parameter of the Laplacian DP noise. 
Our goal is to derive a relationship between the privacy parameter, the significance level (or the probability of false alarm), type II error probability (or the probability of mis-detection) for the attacker to be successful, i.e. to fail to reject $H_0$, as a function of the bias $\Delta \mu$.  
The corresponding likelihood ratio for (\ref{eq:ht}) for the Laplace mechanism is given by
\begin{equation}\label{likelihood_a}
\Lambda = \frac{\mathcal{L}(p_1(\mu_1, b_1); z)}{\mathcal{L}(p_0(\mu_0, b_0); z)} \underset{H_1}{\overset{H_0}{\ou}} \kappa
\end{equation} where $\kappa$ is some positive number to be determined and $(\mu_i, b_i)$ for $i=0,1$ represent the location and scale parameters of the distributions to be tested. 

The next theorem states our first main result.
\begin{theorem}\label{theorem:k_lap}
The threshold of the best critical region of size $\alpha$ defined in (\ref{eq:alpha_def}) for deciding between the null hypothesis and its alternative of the one-sided hypothesis testing problem in (\ref{eq:ht}) for a Laplace mechanism with the largest power $\bar{\beta}$ is given as a function of the probability of false alarm $\alpha$, privacy parameter $\epsilon$ and global sensitivity $s$ as follows
\begin{equation}\label{eq:threshold}
k=
\begin{cases}
\mu_0+\frac{s}{\epsilon} \ln (2(1-\alpha)) & \textrm{if}\; \alpha \in [0, .5]\\
\mu_0-\frac{s}{\epsilon} \ln(2\alpha) & \textrm{if}\; \alpha \in [.5 , 1] 
\end{cases}
\end{equation} Then according to the adversary's hypothesis testing problem, the defender detects the attack for $\Delta \mu >0$ if the output of the Laplace mechanism $Y_0$ exceeds $(k+q(x))$ where $q(.)$ is the noiseless query output. Similarly, for $\Delta \mu<0$, the attack is detected if $Y_0 < q(x)+k$.
\end{theorem} 

\begin{remark} 
The decision rule given by Theorem \ref{theorem:k_lap} is equivalent to comparing the Laplace noise to the threshold $k$ as it will be shown by the following proof. For positive bias, the critical region becomes $(k, \infty)$ thus, $z\underset{H_1}{\overset{H_0}{\ou}} k$.
By analogy if $\Delta \mu <0$, the critical region for the Laplace noise becomes $(-\infty, k)$. %The decision rule given by Theorem \ref{theorem:k_lap} is justified by the assumption that the adversary does not have access to the raw data, rather to the noisy output published by the trusted server.
\end{remark}

\begin{proof}
According to Neyman-Pearson theorem \cite{r6}, each point where $\Lambda \geq \kappa$ composes the best critical region of size $\alpha$ as defined in (\ref{eq:alpha_def}) for this simple hypothesis testing problem. Using the ratio in (\ref{likelihood_a}), we will determine the threshold $k$ as a function of the best critical region, the power of the test, the privacy budget and lastly, the attack. 

We expand $\Lambda$ as follows.
\begin{align} \label{eq:likelihood_a_gen}
\Lambda &=\frac{\frac{1}{2\theta (s/\epsilon)}\exp \left\{-\frac{|z-\mu_1|}{\theta (s/\epsilon)}\right\}}{\frac{1}{2s/\epsilon}\exp \left\{-\frac{|z-\mu_0|}{s/\epsilon}\right\}} \\
&=\frac{1}{\theta} \exp \left\{\frac{\epsilon |z-\mu_0|}{ s}-\frac{\epsilon|z-\mu_1|}{\theta s}\right\}
\end{align} 
The likelihood ratio in (\ref{eq:likelihood_a_gen}) can be summarized by the following piecewise function based on the possible relationships between $\mu_1$ and $z$ due to the absolute value in the exponent of the probability distribution for $\mu_1<\mu_0$.
\begin{equation}\label{eq:lambda_deltamup}
  \Lambda_I= 
  \begin{cases} 
  \frac{1}{\theta} \exp \left\{\frac{\epsilon}{\theta s}\left(z(1-\theta)+\theta \mu_0-\mu_1 \right)\right\}  & \textrm{if}\; z < \mu_1 \\
  \frac{1}{\theta} \exp \left\{-\frac{\epsilon}{\theta s}(z(1+\theta)-\theta \mu_0 - \mu_1)\right\}    & \textrm{if}\; z \in [\mu_1,\mu_0] \\
	\frac{1}{\theta} \exp \left\{-\frac{\epsilon}{\theta s}(z(1-\theta)+\theta \mu_0-\mu_1\right\}    & \textrm{if}\; z \geq \mu_0
  \end{cases}
\end{equation}
Equivalently, $ \Lambda_I$ is confined in the interval $[\frac{1}{\theta} \exp \left\{-\frac{\epsilon}{\theta s}(z(1-\theta)+\theta \mu_0-\mu_1\right\},$ $ \frac{1}{\theta} \exp \left\{\frac{\epsilon}{\theta s}\left(z(1-\theta)+\theta \mu_0-\mu_1 \right)\right\}]$. 
On the other hand, for $\mu_1>\mu_0$, the corresponding likelihood ratio for the hypotheses in (\ref{eq:ht}) yields
\begin{equation}\label{eq:lambda_deltamun}
  \Lambda_{II}= 
  \begin{cases} 
  \frac{1}{\theta}\exp \left\{ \frac{\epsilon}{\theta s}(z(1-\theta)+\theta \mu_0-\mu_1)\right\}  & \textrm{if}\; z < \mu_0 \\
  \frac{1}{\theta}\exp \left\{\frac{\epsilon}{\theta s}(z (1+\theta)-\theta \mu_0- \mu_1)\right\}    & \textrm{if}\; z \in [\mu_0, \mu_1] \\
	\frac{1}{\theta}\exp \left\{-\frac{\epsilon}{\theta s}(z(1-\theta)+\theta \mu_0-\mu_1)\right\}    & \textrm{if}\; z \geq \mu_1
  \end{cases}
\end{equation}
To be able to determine a threshold for deciding between the hypotheses in (\ref{eq:ht}), we compute the false alarm rate $\alpha$ and the mis-detection error $\beta$ (and the power of the test, that is $1-\beta$) applying the Neyman-Pearson lemma that guarantees maximizing the power of the hypothesis test for a given false alarm rate $\alpha$.

\paragraph{Derivation of $\alpha$:} Based on the definition in (\ref{eq:alpha_def}), for $\Delta \mu >0$ the probability of raising a false-alarm is derived by integrating the following probability distribution over the critical region
\begin{align}
\alpha&=\Pr[H_0\;\textrm{reject}|H_0\; \textrm{is}\;\textrm{true}] \\
&= \int_{k}^{\infty} \frac{\epsilon}{2 s} \exp \left\{-\frac{\epsilon|z-\mu_0|}{s}\right\} dz,
\end{align} which is further expanded out in two possible ways. First for $k<\mu_0$, we get
\begin{align}
\alpha&=1-\int_{-\infty}^{k} \frac{\epsilon}{2s}\exp \left\{\frac{\epsilon}{s}(z-\mu_0)\right\} dz \\
%&= 1- \frac{1}{2} \exp \left\{\frac{\epsilon}{s} (z-\mu_0)\right\}\Biggr|_{-\infty}^{k} \\
&=1-\frac{1}{2} \exp \left\{\frac{\epsilon }{s} (k-\mu_0) \right\} \label{eq:alpha_i}
\end{align} 
Second, we have for $k\geq \mu_0$
\begin{align}
\alpha&= \int_{k}^{\infty}\frac{\epsilon}{2s} \exp \left\{-\frac{\epsilon}{s} (z-\mu_0) \right\}dz \\
%&=-\frac{1}{2} \left[\exp\{-\infty\}-\exp\left\{-\frac{\epsilon}{s}(k-\mu_0) \right\}\right] \\
&=\frac{1}{2} \exp\left\{-\frac{\epsilon}{s} (k-\mu_0) \right\} \label{eq:alpha_ii}
\end{align}
Rewriting (\ref{eq:alpha_i}) and (\ref{eq:alpha_ii}) as an equality for $k$, we obtain the piecewise function (\ref{eq:threshold}) as the threshold in Theorem \ref{theorem:k_lap} as a function of $\alpha$. If the bias induced by the adversary is negative, i.e. $\Delta \mu <0$, then the conditions to obtain (\ref{eq:alpha_i}) and (\ref{eq:alpha_ii}) are swapped. For $\Delta \mu <0$ and $k<\mu_0$, we get (\ref{eq:alpha_ii}) for the probability of false-alarm. 

\paragraph{How to determine $\kappa$?:} According to the piecewise expansions of likelihood ratio functions in (\ref{eq:lambda_deltamun}) and (\ref{eq:lambda_deltamup}) respectively for $\Delta \mu <0$ and $\Delta \mu >0$, we have the intervals for $\kappa$ given by (\ref{kappa_int1}) and (\ref{kappa_int2}) on top of the next page since $\Lambda  \underset{H_1}{\overset{H_0}{\ou}} \kappa$. 
\newcounter{tempequationcounter}
\begin{figure*}[h!]
\normalsize
\setcounter{MYtempeqncnt}{\value{equation}}
\setcounter{equation}{20}
\begin{eqnarray}
\frac{1}{\theta} \exp \left\{-\frac{\epsilon}{\theta s}(z(1-\theta)+\theta \mu_0-\mu_1\right\} & < \kappa< & \frac{1}{\theta} \exp \left\{\frac{\epsilon}{\theta s}\left(z(1-\theta)+\theta \mu_0-\mu_1 \right)\right\}, \; \textrm{for}\; \Delta \mu <0 \label{kappa_int1} \\
\frac{1}{\theta}\exp \left\{ \frac{\epsilon}{\theta s}(z(1-\theta)+\theta \mu_0-\mu_1)\right\}&<\kappa <& 
\frac{1}{\theta}\exp \left\{-\frac{\epsilon}{\theta s}(z(1-\theta)+\theta \mu_0-\mu_1)\right\} , \; \textrm{for}\; \Delta \mu >0 \label{kappa_int2} 
\end{eqnarray}
\setcounter{equation}{\value{MYtempeqncnt}}
\hrulefill 
\vspace*{4pt}
\end{figure*}
%\begin{eqnarray}
%\frac{1}{\theta} \exp \left\{\frac{\epsilon}{\theta s}\left(z(1-\theta)+\theta \mu_0-\mu_1 \right)\right\} & < \kappa<& \frac{1}{\theta} \exp \left\{-\frac{\epsilon}{\theta s}(z(1-\theta)+\theta \mu_0-\mu_1\right\} , \; \textrm{for}\; \Delta \mu <0 \\
%\frac{1}{\theta}\exp \left\{ \frac{\epsilon}{\theta s}(z(1-\theta)+\theta \mu_0-\mu_1)\right\} & <\kappa <& \frac{1}{\theta}\exp \left\{-\frac{\epsilon}{\theta s}(z(1-\theta)+\theta \mu_0-\mu_1)\right\} , \; \textrm{for}\; \Delta \mu >0
%\end{eqnarray} 
Therefore, the null hypothesis is rejected for
\addtocounter{equation}{2}
\begin{equation}
\frac{1}{\theta} \exp \left\{-\frac{\epsilon}{\theta s}(z(1+\theta)-\theta \mu_0 - \mu_1)\right\}  < \kappa\; \textrm{for}\; \Delta \mu <0
\end{equation} or
\begin{equation}
\frac{1}{\theta}\exp \left\{\frac{\epsilon}{\theta s}(z (1+\theta)-\theta \mu_0- \mu_1)\right\}> \kappa\; \textrm{for}\; \Delta \mu >0
\end{equation} Due to the threshold of the critical region defined in Theorem \ref{theorem:k_lap}, we finally get $\kappa$ as follows
\begin{align}
\kappa &= \frac{1}{\theta} \exp \left\{-\frac{\epsilon}{\theta s}(k(1+\theta)-\theta \mu_0 - \mu_1)\right\}\; \textrm{for}\; \Delta \mu <0 \\
\kappa &= \frac{1}{\theta}\exp \left\{\frac{\epsilon}{\theta s}(k (1+\theta)-\theta \mu_0- \mu_1)\right\} \;\textrm{for}\; \Delta \mu >0
\end{align} %For the adversary's hypothesis testing, $\kappa$ is shifted as much as $\mu_0$ (?).
\end{proof}
\paragraph{Derivation of the power of the test:} The power of the hypothesis test is the probability of rejecting the null hypothesis $H_0$ given that the alternative hypothesis $H_1$ is true. Let $\bar{\beta}$ denote the complement of the type II error $\beta$, we have using the definition in (\ref{eq:beta_def}) for $\Delta \mu >0$ and $ k< \mu_1 $
\begin{align}\label{eq:beta_i}
\bar{\beta}&=1-\Pr\left[H_1 \textrm{reject}| H_1 \textrm{is}\;\textrm{true}\right]\\
&=\int_{k}^{\infty} \frac{\epsilon}{2\theta s} \exp\left\{\frac{\epsilon (\mu_1-z)}{\theta s}\right\} dz \\
&=\frac{1}{2}\exp\left\{\frac{\epsilon (\mu_1-k)}{\theta s}\right\}
\end{align}
As for $k>\mu_1$, the power function becomes
\begin{align}\label{eq:beta_ii}
\bar{\beta}&=1-\int_{-\infty}^{k} \frac{\epsilon}{2 \theta s} \exp\left\{\frac{\epsilon(z-\mu_1)}{\theta s}\right\} dz\\
&=1-\frac{1}{2}\exp\left\{\frac{\epsilon (k-\mu_1)}{\theta s}\right\}
\end{align} On the contrary for negative bias $\Delta \mu <0$, the conditions based on $k$ and $\mu_1$ to obtain (\ref{eq:beta_i}) and (\ref{eq:beta_ii}) are swapped. In Section \ref{sec:numeric}, we present numerical evaluation results for Theorem 1 using the probability of false-alarm $P_{FA}$ and power of the test $1-P_{MD}=\bar{\beta}$ to draw receiving operating characteristic curves (ROC) as performance analysis.

\begin{remark}
Special case of $\theta=1$ and $|\Delta \mu_1| \leq s$: 
%Depending on the sign of $\mu_1$, we distinguish two possible intervals for $\Lambda$. If $X_a$ is subtracted from the database then the induced bias would be positive. Thus the corresponding expectation is also positive. On the contrary, in case of $X_a$ being an addition on the database then $\mu_1 <0$. 
Setting $\theta=1$ in (\ref{eq:likelihood_a_gen}), it can be easily observed that both likelihood ratios $\Lambda_I$ in (\ref{eq:lambda_deltamup}) and $\Lambda_{II}$ in (\ref{eq:lambda_deltamun}) are included in the following interval $[\exp \left\{-\Delta \mu \frac{\epsilon}{s}\right\}, \exp \left\{\Delta \mu \frac{\epsilon}{s}\right\}]$. 
%Thus, the bias induced by the attack should also be bounded by the sensitivity since the hypothesis test is conducted when one of the parameters of the Laplace distribution is known, i.e. the variance is $s/\epsilon$. 
Applying also $|\Delta \mu_1| \leq s$ onto the likelihood ratio $\Lambda$ in (\ref{eq:likelihood_a_gen}), we get $\exp \{-\epsilon\} \leq \Lambda \leq \exp \{ \epsilon \}$ which is the $(\epsilon,0)-$ DP.
\end{remark}

%%%%%%%%%%%%%%%%%%%%%%%%%%%%%%%%%%%%%%%%%%%%%%%%%%%%%%%%%%%%%%%%%%%%%%%%%%%%%%%%%%%%%%%%%%%%%%%%%%%%%%%%%%%%%%%%%%%%%%%%%
%%%%%%%%%%%%%%%%%%%%%%%%%%%%%%%%%%%%%%%%%%%%%%%%%%%%%%%%%%%%%%%%%%%%%%%%%%%%%%%%%%%%%%%%%%%%%%%%%%%%%%%%%%%%%%%%%%%%%%%%%

\subsection{Two-sided test}

As an alternative solution to the same problem of detecting the attacker through determining the shifts and changes in the location and deviation of the DP noise using a one-sided hypothesis test, a two-sided test could provide a more realistic solution where it is not possible to assume the direction of the shift induced by the adversary. Hence the hypothesis test in (\ref{eq:ht}) can be conducted for determining the (possible) change in the distribution of the DP noise in both directions where the null hypothesis remains the same as 
$H_{0} :Z \sim \mathrm{Lap}(\mu_0, s/\epsilon)$ to test against the alternative $H_1: Z \sim \mathrm{Lap}(\mu_1, \theta s/\epsilon)$. 

This translates to choosing between
\begin{align}%\label{eq:ht_two-sided}
H_0 &: \mu=\mu_0, b= s/\epsilon \label{eq:null_two_sided}\\
H_{1}&:\textrm{at least one of the equalities does not hold}\label{eq:alt_two_sided}
\end{align} where $\mu$ denotes the location parameter and $b$ denoted the scale parameter of any Laplace distribution. The alternative hypothesis can also be stated with the parameters $\mu=\mu_1, b= \theta s/\epsilon$ where $\theta \geq 1$.

In this two-sided test, there are two thresholds on each side of the origin to be determined for the critical region each with a size of $\alpha/2$. Let $k_1$ and $k_2$ denote the threshold greater and smaller than the origin, respectively. The next theorem presents the thresholds for detecting the attack as a function of the probability of false-alarm.
\begin{theorem}\label{theorem_threshold_two-sided}
The threshold of the best critical region of size $\alpha$ defined in (\ref{eq:alpha_def}) for choosing between the null hypothesis and its alternative of the two-sided hypothesis testing problem in (\ref{eq:null_two_sided})-(\ref{eq:alt_two_sided}) for a Laplace mechanism with the largest power $\bar{\beta}$ is
\begin{eqnarray}
k_1&=&\mu_0 -(s/\epsilon) \log \alpha \label{eq:thresholds_twosidedI}\\
k_2&=&\mu_0+(s/\epsilon) \log \alpha \label{eq:thresholds_twosidedII}
\end{eqnarray} Then according to the adversary's hypothesis testing problem, the defender fails to detect the attack when the output of the Laplace mechanism $Y_0$ is confined in $(q(x)+k_2, q(x)+k_1)$ where $q(.)$ is the noiseless query output.
\end{theorem} 

\begin{proof} The null hypothesis cannot be rejected if the noisy output of the Laplace mechanism confined in the interval $(k_2, k_1)$. % for some $\kappa_2 <\mu_0$ and $\kappa_1 > \mu_0$. 
First, we begin with the derivation of threshold for the output of the DP mechanism. 
The probability of raising a false-alarm or having a type I error is derived as follows.
\begin{align}
\alpha &=\Pr[H_0\;\textrm{reject}|H_0\; \textrm{is}\;\textrm{true}] \\
&= \int_{-\infty}^{k_2} \frac{\epsilon}{2 s} \exp \left\{\frac{\epsilon (z-\mu_0)}{s}\right\} dz \label{eq:alpha-twosidedI}\\
&+\int_{k_1}^{\infty} \frac{\epsilon}{2 s} \exp \left\{-\frac{\epsilon (z-\mu_0)}{s}\right\} dz \label{eq:alpha-twosidedII}
\end{align} 
Each addend of $\alpha$ corresponds to one half of the probability of false-alarm. Equating each integral to $\alpha/2 $ and rewriting the equalities in terms of $k_1$ and $k_2$, we get the thresholds (\ref{eq:thresholds_twosidedI}) and (\ref{eq:thresholds_twosidedII}).
\end{proof}
%%%%%%%%%%%%%%%%%%%%%%%%%%%%%%%%%%%%%%%%%%%%%%%%%%%%%%%%%%%%%%%%%%%%%%%%%%%%%%%%%%%%%%%%%%%%%%%%%%%%%%%%%%%%%%%%%%%%%%%%%
%%%%%%%%%%%%%%%%%%%%%%%%%%%%%%%%%%%%%%%%%%%%%%%%%%%%%%%%%%%%%%%%%%%%%%%%%%%%%%%%%%%%%%%%%%%%%%%%%%%%%%%%%%%%%%%%%%%%%%%%%

\subsubsection*{A trade-off between $\mu_1$, $s$ and $\epsilon$ for detecting the attacker-Two-sided test}

On the defender's end, the DP mechanism wants to detect the attacker to preserve DP. Using the threshold presented in Theorem \ref{theorem_threshold_two-sided}, we can determine an interval to confine the mean of the attacker's advantage to be detected by the DP mechanism, i.e. for the null hypothesis $H_0$ to be rejected. Alternatively, such an interval can be converted for the privacy parameter $\epsilon$ as a function of error probabilities, the attack and the sensitivity. The following result, Corollary \ref{corr_twosided}, presents upper and lower bounds on the attacker's advantage so that the defender detects the attack. 

There are two possible cases w.r.t. the relationship between $\mu_0$ and $\mu_1$. The alternative hypothesis in this two-sided test also states that these two parameters are unequal. As we have discussed earlier in the derivation of the threshold for determining the critical region in Laplace mechanisms, whether $\mu_0>\mu_1$ or $\mu_1>\mu_0$ directly effects the likelihood ratio function, and thus the condition to reject the null hypothesis. Let us the consider the first possible case of $\mu_0<\mu_1$. In this case, we have either $k_2<\mu_0< k_1 <\mu_1$ or $k_2<\mu_0<\mu_1< k_1 $. On the contrary for $\mu_1<\mu_0$, we have for the thresholds either of the cases $\mu_1<k_2<\mu_0<k_1$ or $k_2<\mu_1<\mu_0<k_1$. These different cases can be used for deriving an interval to include $\Delta \mu$ as a function of the error probabilities, privacy budget and the sensitivity.

\begin{corollary} \label{corr_twosided}
The absolute bias $|\Delta \mu|=|\mu_1-\mu_0|$ induced by the adversary is confined in the following interval so that the defender detects $X_a$ and preserves $(\epsilon, 0)$- DP
\begin{equation} \label{interval_twos}
\frac{s}{\epsilon}\log \left(\alpha \bar{\beta}^{\theta} \right) < \Delta \mu < \frac{s}{\epsilon}\log \left(\frac{1}{\alpha \bar{\beta}^{\theta}}\right) 
\end{equation} for $\theta \geq 1$ where $\alpha$ and $\bar{\beta}$ respectively are the significance level and the power of the test of (\ref{eq:alt_two_sided}).
\end{corollary} 
\begin{proof}
We begin with deriving the power of the two-sided test (\ref{eq:alt_two_sided}) as a function of the thresholds of the critical region.
\paragraph{Derivation of the power of the test:} The probability of correctly detecting the attacker as follows.
\begin{align}
\bar{\beta}&=\int_{-\infty}^{k_2} \frac{\epsilon}{2 \theta s} \exp \left\{\frac{\epsilon (z-\mu_1)}{\theta s}\right\} dz  \label{line1}\\
&+\int_{k_1}^{\infty} \frac{\epsilon}{2 \theta s} \exp \left\{-\frac{\epsilon (z-\mu_1)}{\theta s}\right\} dz \label{line2}\\
&=\frac{1}{2} \exp\left\{\frac{\epsilon(k_2-\mu_1)}{\theta s}\right\}+\frac{1}{2} \exp\left\{-\frac{\epsilon(k_1-\mu_1)}{\theta s}\right\} \label{eq:barbeta_two}
\end{align}
Each addend in (\ref{line2}), correspond to $\bar{\beta}/2$ and can be rewritten for the thresholds as functions of the power of the test as $k_1=\mu_1-\frac{s}{\epsilon}\log (\bar{\beta})^{\theta}$ and $k_2=\mu_1+\frac{s}{\epsilon}\log (\bar{\beta})^{\theta}$. Combining this with $k_2<k_1$ for the case $\mu_0<\mu_1$, the bias is lower bounded as follows
\begin{align}
\mu_0+\frac{s}{\epsilon} \log \alpha &< \mu_1-\frac{\theta s}{\epsilon}\log \bar{\beta}\\
\frac{s}{\epsilon} \log \left(\alpha \bar{\beta}^{\theta}\right) & < \Delta \mu \label{eq:lb_pos}
\end{align} As for the upper bound we have
\begin{align}
\mu_1+\frac{\theta s}{\epsilon} \log \bar{\beta} &< \mu_0-\frac{s}{\epsilon}\log \alpha\\
 \Delta \mu & < \frac{s}{\epsilon} \log \left(\frac{1}{\alpha \bar{\beta}^{\theta}} \right) \label{eq:lb_neg}
\end{align} 
By analogy, we have get the swapped upper and lower bound for $-\Delta \mu $ for the second case of $\mu_1<\mu_0$. 
Finally, we get the interval for the absolute bias as given by (\ref{interval_twos}). This concludes the proof of the corollary.
\end{proof}

%%%%%%%%%%%%%%%%%%%%%%%%%%%%%%%%%%%%%%%%%%%%%%%%%%%%%%%%%%%%%%%%%%%%%%%%%%%%%%%%%%%%%%%%%%%%%%%%%%%%%%%%%%%%%%%%%%%%%%%%%
%%%%%%%%%%%%%%%%%%%%%%%%%%%%%%%%%%%%%%%%%%%%%%%%%%%%%%%%%%%%%%%%%%%%%%%%%%%%%%%%%%%%%%%%%%%%%%%%%%%%%%%%%%%%%%%%%%%%%%%%%

\section{Relative Entropy \label{sec:KL}}

This section is dedicated to the derivation of relative entropy or Kullback-Leibler (KL) divergence between two Laplace distributions and its adaptation to adversarial classification through \textit{KL-DP}. 
\begin{definition}[KL-DP, \cite{r2}]
For a randomized mechanism $P_{Y|X}$ that guarantees $\epsilon-$ KL-DP, if the following inequality holds for all its neighboring datasets $x$ and $\tilde{x}$.
\begin{equation} \label{eq:KLdp}
D(P_{Y|X=x}||P_{Y|X=\tilde{x}}) \leq \exp\{\epsilon\}
\end{equation}
\end{definition}
In \cite[Theorem 1]{r2}, KL-DP is proven to satisfy the following chain of inequalities
\begin{equation}
(\epsilon, 0)-\textrm{DP} \geq \textrm{KL}-\textrm{DP} \geq (\epsilon, \delta)-\textrm{DP}
\end{equation}
For the described problem and the associated model described in Section \ref{subsec:model}, the neighboring datasets could be imagined as those where the output of the query is $\sum_{i=1}^n X_i$ before the attack and $(\sum_{i=1}^n X_i + X_a)$ after the attack. 
%It should be noted that the adversary is able to modify a single entry of the dataset hence the size of the dataset is assumed not to have changed due to the attack. In this case, 
The corresponding distributions are considered as the DP noise with and without the induced value of $X_a$ by the attacker as in our original hypothesis testing problem in (\ref{eq:ht}). To be consistent with the hypotheses in (\ref{eq:ht}), we set $P_{Y|X=x} ~ Lap(\mu_0, s/\epsilon)$ and for the neighbor, we have $ Lap(\mu_1, \theta s/\epsilon)$. %Note that the relative entropy between two Laplacians is not symmetrical.

Hereafter, we derive the relative entropy between $p_0 \sim Lap (\mu_0, b_0)$ and $p_1 \sim Lap (\mu_1, b_1)$.
\begin{align}
&D(p_0||p_1)= \int  p_0(z) \log\frac{p_0(z)}{p_1(z)} dz \\
%&=\mathbb{E}_{p_0}\left[\log \frac{p_0}{p_1} \right] \\
&=\mathbb{E}_{p_0} \left[\log \frac{1/2b_0\exp\left\{-\frac{|z-\mu_0|}{b_0}\right\}}{1/2b_1 \exp \left\{-\frac{|z-\mu_1|}{b_1}\right\}}\right]\\
%&=\mathbb{E}_{p_0} \left[\log \left( \frac{b_1}{b_0} \right) - \frac{|z- \mu_0|}{b_0} + \frac{|z-\mu_1|}{b_1} \right] \\
&=\log\left(\frac{b_1}{b_0}\right)-\frac{1}{b_0} \mathbb{E}_{p_0}[|z-\mu_0|]+\frac{1}{b_1} \mathbb{E}_{p_0}[|z-\mu_1|]\\
&\overset{(a)}{=}\log \left(\frac{b_1}{b_0}\right)-1+\frac{1}{b_1}\mathbb{E}_{p_0}[|z-\mu_1|] \label{eq:(a)}
\end{align}

In step (a), we substituted $\mathbb{E}_{p_0}[|z-\mu_0|]$ by $b_0$ since for $z\sim Lap(\mu,b)$ then $|z-\mu|\sim Exp(1/b)$ and the corresponding mean for the exponential random variable is the inverse of its parameter. For the last term, $\frac{1}{b_1}\mathbb{E}_{p_0}[|z-\mu_1|]$, we must consider two different cases due to the absolute value in the exponent of the Laplace distribution. In the following first expansion, the two distributions are centered around $\mu_0$ and $\mu_1$ where $\mu_0<\mu_1$. 
\begin{align}
\frac{1}{b_1}\mathbb{E}_{p_0}[|z-\mu_1|]&=\frac{1}{b_1} \int_{p_0} |z-\mu_1|\frac{1}{2 b_0}\exp \left \{-\frac{|z-\mu_0|}{b_0} \right\}dz\\
%&=\frac{1}{b_1} \int_{-\infty}^{\mu_0} (-z+\mu_1) \frac{1}{2 b_0}\exp \left \{\frac{z-\mu_0}{b_0} \right\} dz\\
%&+\frac{1}{b_1} \int_{\mu_0}^{\mu_1} (-z+\mu_1) \frac{1}{2 b_0}\exp \left \{\frac{-z+\mu_0}{b_0} \right\} dz\\
%&+\frac{1}{b_1} \int_{\mu_1}^{\infty} (z-\mu_1) \frac{1}{2 b_0}\exp \left \{\frac{-z+\mu_0}{b_0} \right\}dz \\
&=\frac{b_0}{2b_1} \exp \left\{\frac{\mu_0-\mu_1}{b_0}\right\} + \frac{\mu_1-\mu_0}{b_1}\label{eq:lastterm_1}
\end{align}
Substituting (\ref{eq:lastterm_1}) into (\ref{eq:(a)}), we finally get the KL divergence between two Laplacians as
\begin{equation}\label{eq:KL_laplace}
D_{I}(p_0||p_1)= \log \left(\frac{b_1}{b_0}\right)-1 + \frac{b_0}{b_1} \exp \left\{\frac{\mu_0-\mu_1}{b_0}\right\} - \frac{\mu_0-\mu_1}{b_1}
\end{equation}
Simplifying (\ref{eq:KL_laplace}) for $b_0=s/\epsilon$ and $b_1=\theta (s /\epsilon)$ and $\mu_1-\mu_0=\Delta \mu$ for the hypothesis testing problem defined in (\ref{eq:ht}), we finally get
\begin{equation}\label{eq:KL_equalscale}
D_{I}(p_0||p_1)_{\Delta \mu>0}= \log \theta -1 +\frac{1}{\theta} \exp\left\{-\frac{\Delta \mu \epsilon}{s}\right\}+\frac{\Delta \mu \epsilon}{\theta s}
\end{equation}

As for the case of $\mu_0>\mu_1$, we have
\begin{align}\label{eq:lastterm_2}
\frac{1}{b_1}\mathbb{E}_{p_0}[|z-\mu_1|]&=\frac{1}{b_1} \int_{p_0} |z-\mu_1|\frac{1}{2 b_0}\exp \left \{-\frac{|z-\mu_0|}{b_0} \right\}dz\\
%&=\frac{1}{b_1} \int_{-\infty}^{\mu_1} (-z+\mu_1) \frac{1}{2 b_0}\exp \left \{\frac{z-\mu_0}{b_0} \right\}\\
%&+\frac{1}{b_1} \int_{\mu_1}^{\mu_0} (z-\mu_1) \frac{1}{2 b_0}\exp \left \{\frac{z-\mu_0}{b_0} \right\} dz\\
%&+\frac{1}{b_1} \int_{\mu_0}^{\infty} (z-\mu_1) \frac{1}{2 b_0}\exp \left \{\frac{-z+\mu_0}{b_0} \right\}dz \\
&=\frac{b_0}{b_1} \exp \left\{\frac{\mu_1-\mu_0}{b_0}\right\} + \frac{b_0}{b_1}
\end{align} Finally, plugging (\ref{eq:lastterm_2}) into (\ref{eq:(a)}), we get the KL divergence between $p_0$ and $p_1$ for positive $\mu_0$ where $\mu_0>\mu_1$ as follows.
\begin{equation}\label{eq:KL_laplaceII}
D_{II}(p_0||p_1)= \log \left(\frac{b_1}{b_0}\right)-1 + \frac{b_0}{b_1} \exp \left\{\frac{\mu_1-\mu_0}{b_0}\right\} + \frac{b_0}{b_1}
\end{equation}
Setting $\mu_1-\mu_0=\Delta \mu$, $b_0=s/\epsilon$ and $b_1=\theta (s /\epsilon)$ in (\ref{eq:KL_laplaceII}), $D_{II}(p_0||p_1)$ yields
\begin{equation}\label{eq:KL_equalscaleII}
D_{II}(p_0||p_1)_{\Delta \mu<0}= \log \theta -1+\frac{1}{\theta}\exp\left\{-\frac{\epsilon \Delta \mu}{s}\right\} + \frac{1}{\theta}
\end{equation} %where $\gamma= 1$ since $\mu_0>0$.
\begin{remark}
Authors of \cite{r1} also seek for the maximum bias induced by the adversary where the objective function is the minimum relative entropy between the probability distribution of the dataset before ($p_0$) and after the attack ($p_1$). Nevertheless, the choice of the objective function is set as $D(p_1||p_0) \leq \gamma$ for some $\gamma$. 
%based on Stein's lemma which states that the best error exponent in a hypothesis testing problem choosing between two distributions converges to the relative entropy that is . 
%But the multi-criteria optimization problem of \cite{r1} uses the relative entropy $D(p_1||p_0)$ as the objective function. 
For the Laplace distribution, KL divergence is not symmetric, hence $D(p_0||p_1) \neq D(p_1||p_0)$. Therefore, due to Stein's lemma \cite{r14}, (\ref{eq:KL_laplace}) and (\ref{eq:KL_laplaceII}) should be used instead.
\end{remark}

In Section \ref{sec:numeric}, we present numerical evaluation results of (\ref{eq:KL_equalscale}) for different values privacy parameter as well as various levels of attack.

%<<<<<<< .mine
%\input{numerical}
\section{Numerical Evaluation \label{sec:numeric}}
\paragraph{ROC curves for hypothesis tests:}
Figure \ref{fig:means_laplace_onesided} presents the numerical evaluation results of the one-sided hypothesis test for the Laplace DP noise parameters. The plots depict three different possible scenarios where the induced bias by the adversary is above, equal to and below the sensitivity of the system. $\mu_0$ is taken equal to 0 hence $\Delta \mu= \mu_1$. As highlighted in Figure \ref{fig:means_laplace_onesided}, we plot the ROC curve under different $\epsilon$ and $\theta$. We observe that when the privacy parameter $\epsilon$ is very small (e.g., $\epsilon=0.015$), the test is no longer accurate and detecting the adversary can be considered similar to random guessing. On the other hand, when the privacy parameter is very large, the accuracy of the test becomes higher in the expense of the privacy guarantee. Furthermore,  as opposed to \cite[Theorem 5]{r3}, we notice that ROC curves strongly depend on the sensitivity $s$, hence the mapping function (query) applied on the input. Indeed, when $\mu_1>s$ the accuracy of the test becomes less important as the adversary is trying to harm the system. Figure \ref{fig:means_laplace_onesided} also shows that the choice of $\theta$ affects the power of the test. When $\theta=1$, the test only consists in choosing between two location parameters.
W.r.t. to the choice of $\theta$, Figure \ref{fig:means_laplace_onesided} shows that the power of the test on the y-axis decreases with $\theta$. For each value of $\epsilon$, ROC curves that correspond to $\theta=1$ outperform those with bigger variance as of a certain level of $\alpha$ and as the privacy is decreased (equivalently $\epsilon$ is increased) this flip in performance occurs for much smaller choices of the probability of false alarm. 
%To conclude, when privacy decreases (i.e. $\epsilon$ is large)

The ROC curves corresponding to two-sided hypothesis test (\ref{eq:alt_two_sided}) are depicted in Figure \ref{fig:means_laplace_twosided} for same values of privacy budget and $\theta$ used in the previous case. As expected, ROC curves for the two-tailed test show the same behavior as in Figure \ref{fig:means_laplace_onesided} w.r.t. the effect of the change in the privacy budget on the accuracy of the test ($\bar{\beta}$ increases with $\epsilon$). On the other hand, we observe that in the second case the test is less accurate. This is justified by the lack of knowledge on the sign of $\Delta \mu$. Indeed, the previous test is considered as being more precise ($\Delta \mu >0$). 

\begin{figure}[ht!]
 \centering
 \begin{minipage}[b]{1\linewidth}
 \includegraphics[width=1.0\linewidth]{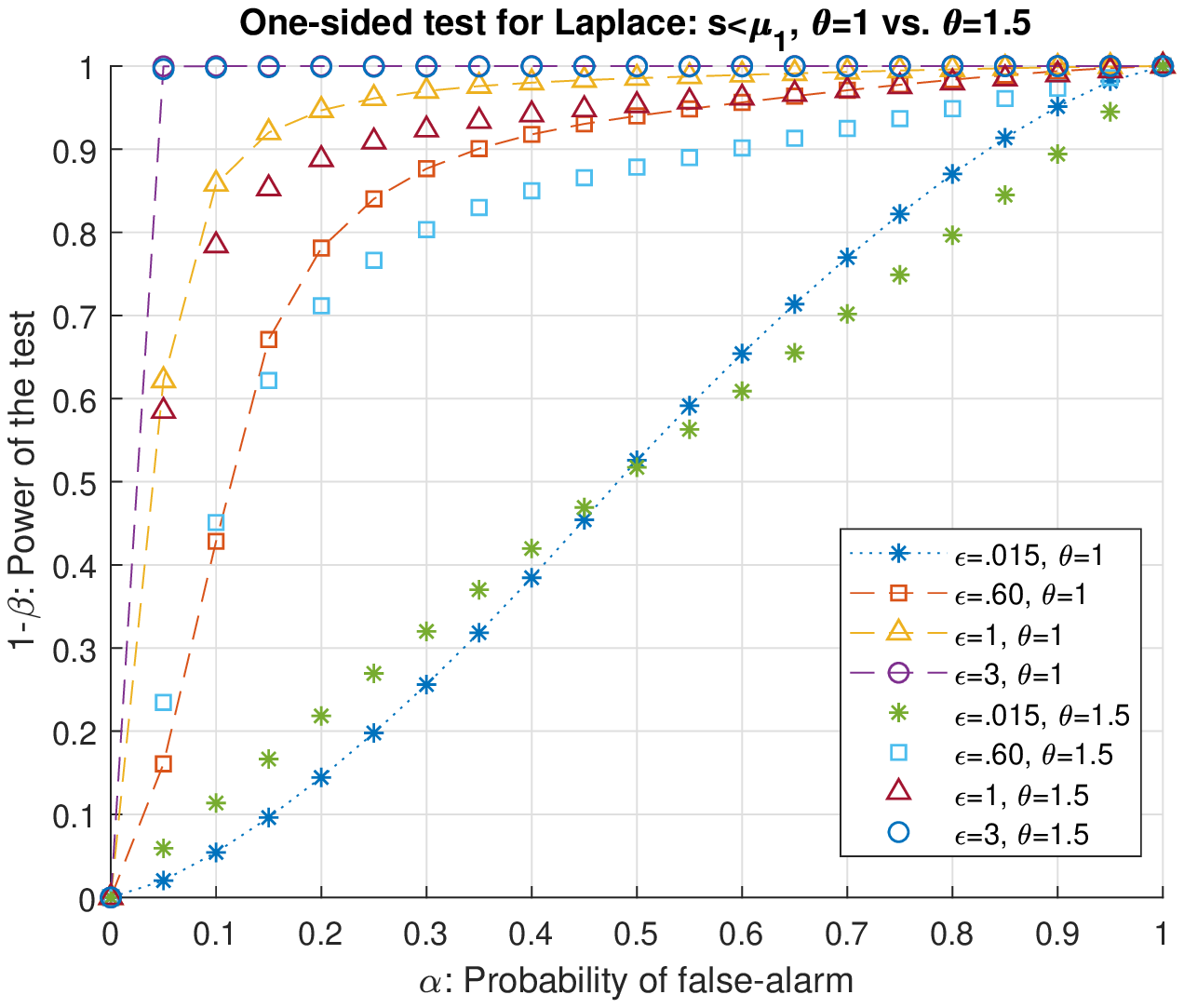}
 \end{minipage}
  \quad
 \begin{minipage}[b]{1\linewidth}
 \centering
 \includegraphics[width=1.0\linewidth]{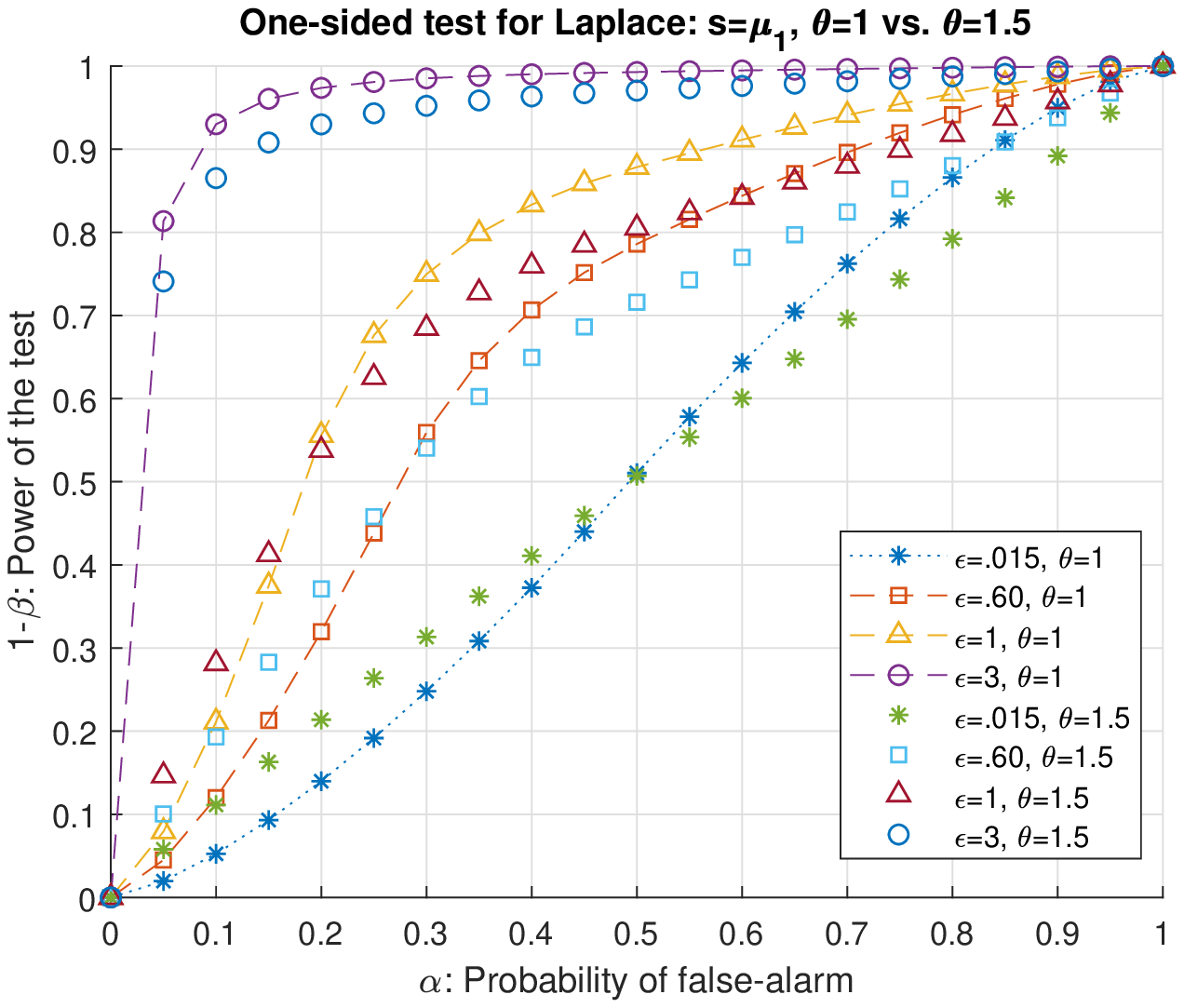}
 \end{minipage}
 \quad
 \begin{minipage}[b]{1\linewidth}
 \centering
 \includegraphics[width=1.0\linewidth]{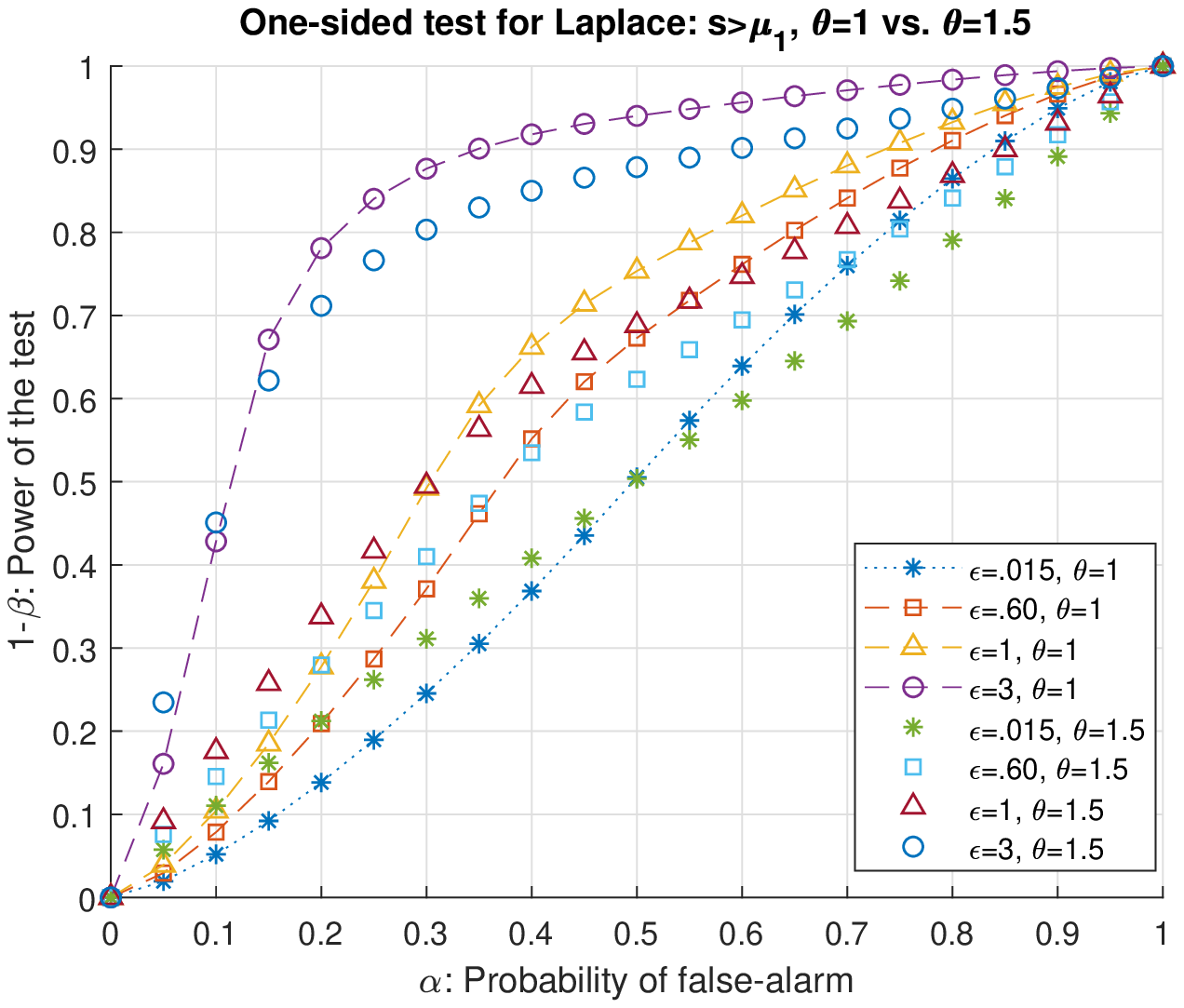}
 \end{minipage}
 \caption{ROC curves for the one-sided hypothesis test ($\Delta \mu= \mu_1 >0$): (\ref{eq:alpha_i}) and (\ref{eq:alpha_ii}) vs. (\ref{eq:beta_i})-(\ref{eq:beta_ii}) for different values of privacy parameter and $\theta=1.5$.} 
\label{fig:means_laplace_onesided}
 \end{figure}
\begin{figure}[ht!]
 \centering
 \begin{minipage}[b]{1\linewidth}
 \includegraphics[width=1.0\linewidth]{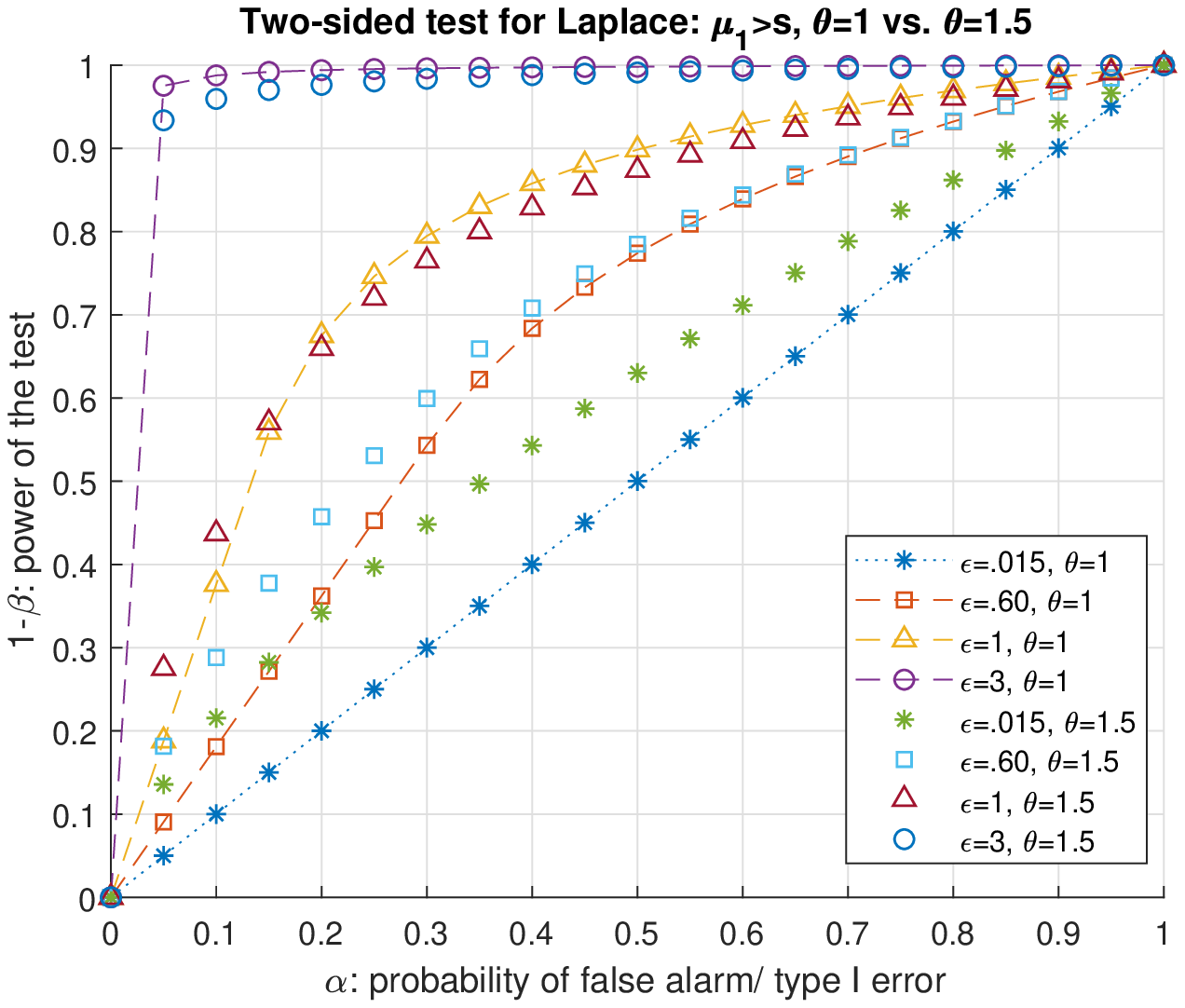}
 \end{minipage}
  \quad
 \begin{minipage}[b]{1\linewidth}
 \centering
 \includegraphics[width=1.0\linewidth]{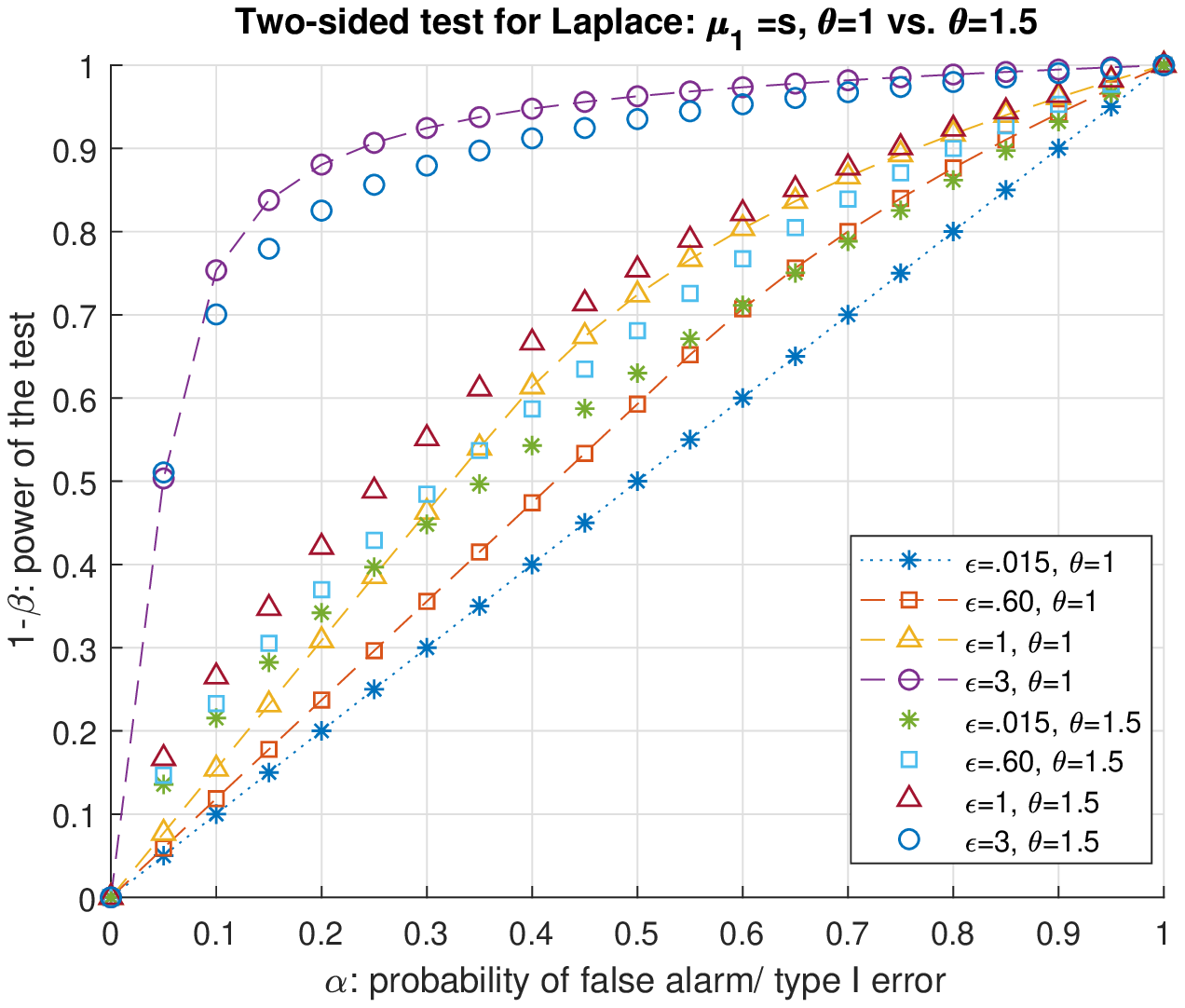}
 \end{minipage}
 \quad
 \begin{minipage}[b]{1\linewidth}
 \centering
 \includegraphics[width=1.0\linewidth]{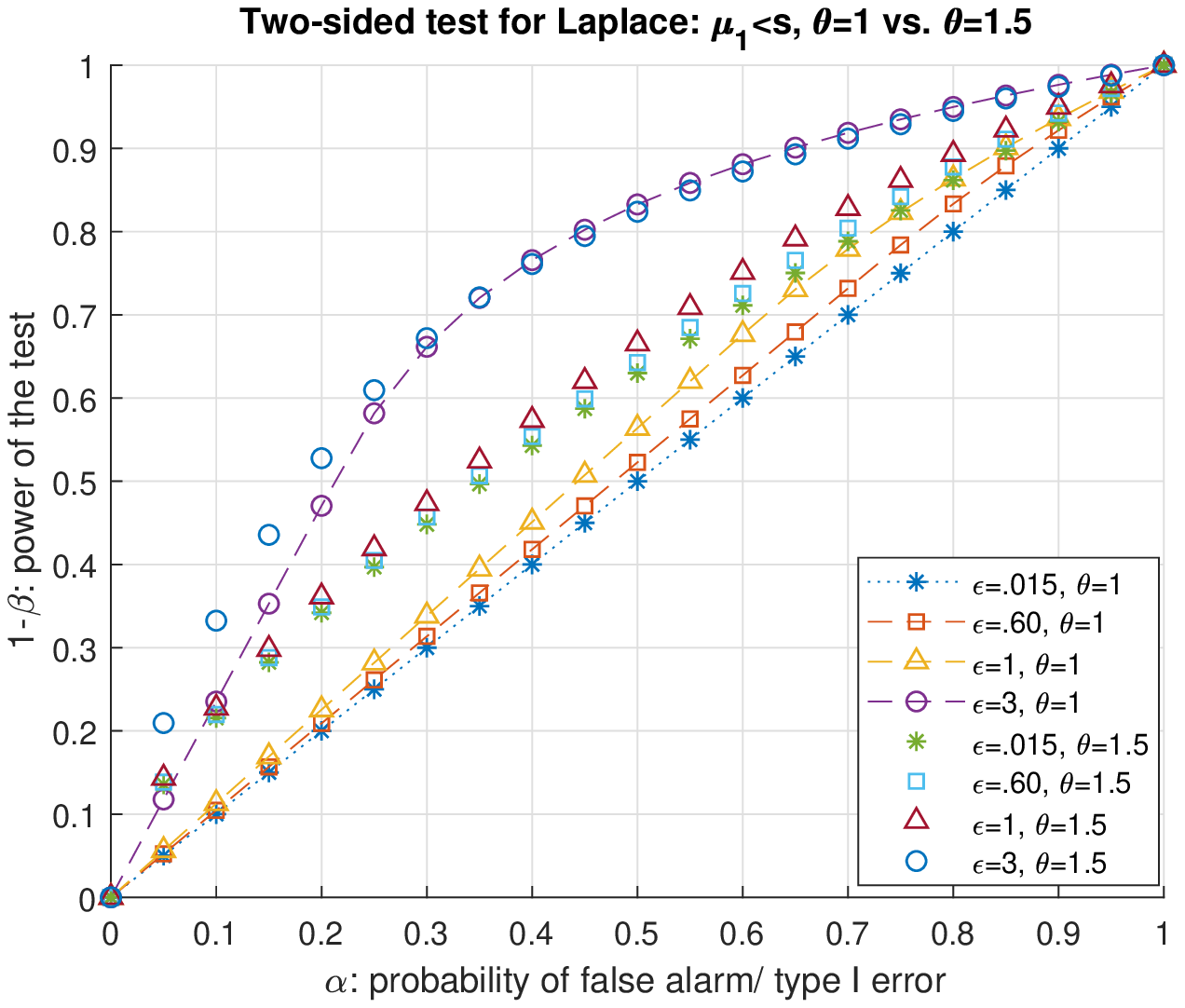}
 \end{minipage}
 \caption{ROC curves for the two-sided hypothesis test ($\Delta \mu= \mu_1 >0$): (\ref{eq:alpha-twosidedI}) and (\ref{eq:alpha-twosidedII}) vs. (\ref{eq:barbeta_two}) depicted for different values of privacy parameter and $\theta=1$ vs $\theta=1.5$.} 
\label{fig:means_laplace_twosided}
 \end{figure}
%\end{comment}
\paragraph{KL-DP vs. privacy budget:}
KL-DP (\ref{eq:KL_equalscale}) derived in Section \ref{sec:KL} is numerically evaluated for different levels of attack in comparison to the sensitivity of the system for both $\theta=1$ and $\theta=1.5$ in Figure (\ref{fig:KL-lap}). Accordingly, the effect of the attack is compared with the upper bound $\exp\{\epsilon\}$ (\ref{eq:KLdp}). Figure (\ref{fig:KL-lap}) shows that increasing the impact the attack w.r.t. the sensitivity, closes the gap with the upper bound and for the case $|\Delta\mu|= 4*s$ and for moderate privacy budget, KL-DP upper bound is violated. 

\begin{figure}[ht!]
 \centering
 %\begin{minipage}[b]{1\linewidth}
 \includegraphics[width=1.0\linewidth]{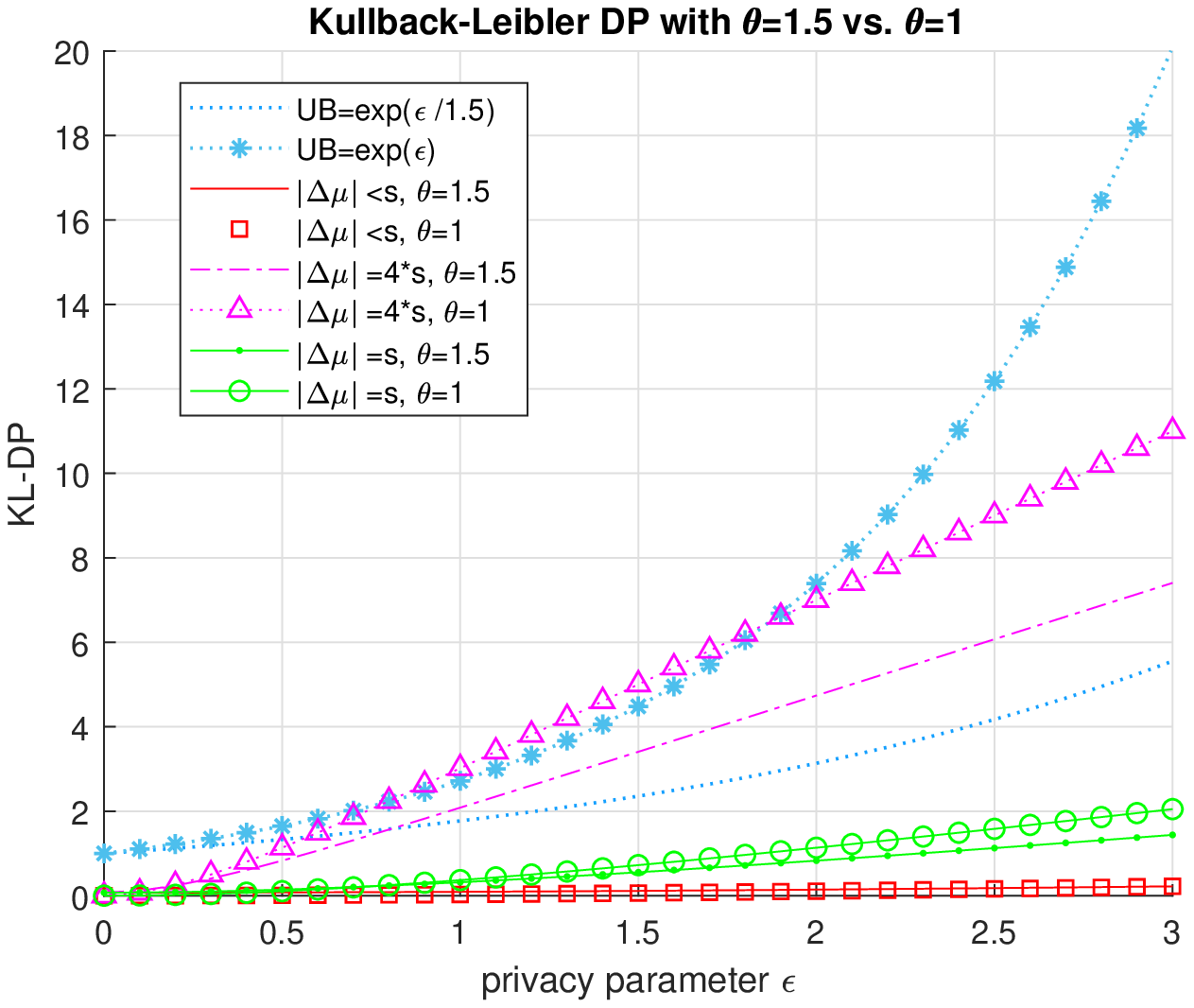}
% \end{minipage}
  %\quad
 %\begin{minipage}[b]{1\linewidth}
 %\centering
 %\includegraphics[width=1.0\linewidth]{plots/theta1bucuk_KL_lap.eps}
 %\end{minipage}
 \caption{KL-DP for different values of privacy parameter and $\theta=1$ vs $\theta=1.5$.} 
\label{fig:KL-lap}
 \end{figure}

\section{Conclusions and Future Work \label{sec:conc}}
 
We characterized a statistical trade-off between the security of the Laplace mechanism and the privacy protected adversary's advantage for adversarial classification using one and two-tailed hypothesis testing. Particularly, in both settings, we established trade-offs between the sensitivity of the system, privacy parameter and the bias due to the attack using the threshold(s) of the critical region in choosing between the hypotheses whether or not the defender detects the attack. Such trade-offs are presented as functions of corresponding error probabilities. Numerical evaluation results show that increasing the privacy parameter also increases the accuracy of the hypothesis test. Additionally, we derived KL-DP for adversarial classification. According to the numerical evaluation results, the effect of increasing the impact of the attack closes the gap with the DP upper bound $\exp\{\epsilon\}$ and some even violates it for moderate privacy budget. 

The authors wish to continue this line of work on different differentially private mechanisms (e.g. Gaussian, Exponential) as well as different type of query functions (e.g. neural network prediction) to understand its effect on the trade-off through the sensitivity.

\bibliography{paper_adv_dp}
\bibliographystyle{IEEEtran}

\end{document}